\documentclass[11pt]{amsart}
\usepackage{amsmath}
\usepackage{amsfonts}
\usepackage{amssymb}
\newtheorem{thm}{Theorem}

\newtheorem{cor}[thm]{Corollary}

\newtheorem{lemma}[thm]{Lemma}

\newtheorem{exam}[thm]{Example}
\newtheorem{defn}[thm]{Definition}

\newcommand{\bra}[1]{\langle #1 |}
\newcommand{\ket}[1]{| #1 \rangle}

\newcommand{\ketbra}[2]{| #1 \rangle\langle #2 |}
\newcommand{\Tr}{{\rm Tr}}
\newcommand{\bb}[1]{\mathbb{#1}}
\newcommand{\cl}[1]{\mathcal{#1}}

\newcommand{\fA}{\mathfrak{A}}
\newcommand{\fB}{\mathfrak{B}}
\begin{document}

\title[Generalized Multiplicative Domains]{Generalized Multiplicative Domains and \\ Quantum Error Correction}

\author[N.~Johnston, D.~W.~Kribs]{Nathaniel Johnston$^1$ and David~W.~Kribs$^{1,2}$}
\address{$^1$Department of Mathematics \& Statistics, University of Guelph,
Guelph, ON, Canada N1G 2W1}
\address{$^2$Institute for Quantum Computing, University of Waterloo, Waterloo, ON, Canada
N2L 3G1}

\begin{abstract}
Given a completely positive map, we introduce a set of algebras that we refer to as its generalized multiplicative domains.
These algebras are generalizations of the traditional multiplicative domain of a completely positive map and we derive a characterization of them in the unital, trace-preserving case, in other words the case of unital quantum channels, that extends Choi's characterization of the multiplicative domains of unital maps. We also derive a characterization that is in the same flavour as a well-known characterization of bimodules, and we use these algebras to provide a new representation-theoretic description of quantum error-correcting codes that extends previous results for unitarily-correctable codes, noiseless subsystems and decoherence-free subspaces.
\end{abstract}

\maketitle

%%%%%%%%%%%%%%%%%%%%%%%%%%%%%%%%%%%%%%%%%%%%%%%%%%%%%%%%%%%
\section{Introduction}
%%%%%%%%%%%%%%%%%%%%%%%%%%%%%%%%%%%%%%%%%%%%%%%%%%%%%%%%%%%

The multiplicative domain of a completely positive map was first studied in operator theory over thirty years ago \cite{Cho74}, and very recently it was discovered that it plays a role in quantum error correction \cite{CJK08}. One of the most important and well-known results about the multiplicative domain says that if the given map is unital, then it can be characterized internally; that is, entirely in terms of the action of the map on the elements of the multiplicative domain \cite{Cho74,Paulsentext}.

In this paper, we introduce a natural generalization of the multiplicative domain to a set of algebras that we refer to as the {\em generalized multiplicative domains} of the given map. We show that, when the map is both unital and trace-preserving, each of these algebras has an internal characterization analogous to that of the standard multiplicative domain, albeit one that is slightly more delicate. Interestingly, the generalized internal characterization does not hold if trace-preservation is removed from the hypotheses. We note that generalized multiplicative domains in certain cases are bimodules, and we present a second internal characterization that is analogous to a known result that characterizes when a completely positive map is a bimodule map for a given subalgebra. We also present a third internal characterization not motivated by previous work in operator theory, but rather by quantum error correction.

Quantum error correction is one of the most important and active areas of research in quantum information \cite{Got02,KLABVZ02,NC00}, and our contribution includes a new representation theoretic description of subspace and subsystem
codes \cite{AlyKla07,Bac05,KlaSar06,KL95,KLP05,KLPL06,Pou05,ShaLid05}. On one hand, these characterizations of correctable codes provide motivation for some of the characterizations of generalized multiplicative domains that are derived. On the other hand, we feel they are of independent interest, as
they generalize several known results for noiseless subsystems and decoherence-free subspaces \cite{duan97,kempe01,Kni06,knill00,lidar98,palma96,zanardi01a,zanardi97}, including the multiplicative domain results of \cite{CJK08} and the commutant relationships of \cite{CK06,HKL04}. In particular, it was recently shown that the standard multiplicative domain of a completely positive trace-preserving (CPTP) map encodes a subclass of what are known as ``unitarily correctable codes'' \cite{CJK08,KLPL06,KS06,SMKE08} (UCC) -- we show that the generalized multiplicative domains of a CPTP map encode {\em all} of its correctable codes. Our results complement other recently-obtained descriptions of subsystem codes \cite{BKK08,BKP08,BKNPV08,CJK08,Kni06,KS06}.

In Section~\ref{sec:prelim} we will introduce the necessary mathematical notation and preliminaries. In Section~\ref{sec:multDom} the multiplicative domain, bimodules, and generalized multiplicative domains will be defined and explored. In Section~\ref{sec:qec} we will explore the connection between the generalized multiplicative domains of a map and its correctable codes.

%%%%%%%%%%%%%%%%%%%%%%%%%%%%%%%%%%%%%%%%%%%%%%%%%%%%%%%%%%%
\section{Preliminaries}\label{sec:prelim}
%%%%%%%%%%%%%%%%%%%%%%%%%%%%%%%%%%%%%%%%%%%%%%%%%%%%%%%%%%%

We will use $\mathcal{H}$ to denote a finite-dimensional Hilbert space and $\mathcal{L}(\mathcal{H})$ to denote the set of linear
operators on $\mathcal{H}$. Given a completely positive (CP) map $\phi : \mathcal{L}(\mathcal{H})
\rightarrow \mathcal{L}(\mathcal{H})$, we know that a family of operators $\phi\equiv\{E_i\}$ can be found such that
$\phi(a) = \sum_i E_i a E_i^*$ for all
$a \in\mathcal{L}(\mathcal{H})$. A CP map is {\em unital} if $\phi(I) = I$ (i.e., $\sum_i E_i E_i^* = I$) and it is {\em trace-preserving} if $\Tr(\phi(a)) = \Tr(a)$ $\forall \, a \in \cl{L}(\cl{H})$ (i.e., $\sum_i E_i^* E_i = I$). We will abbreviate CP trace-preserving maps as CPTP for brevity. CPTP maps are often referred to as {\em quantum channels} or \emph{quantum operations} in quantum information studies \cite{NC00}.

Given a CP map $\phi\equiv\{E_i\}$, we can define its dual map
$\phi^\dagger : \mathcal{L}(\mathcal{H}) \rightarrow
\mathcal{L}(\mathcal{H})$ via $\Tr(\phi(a)b) =
\Tr (a \phi^\dagger(b))$. Observe that
$\phi\equiv\{E_i\}$ if and only if
$\phi^\dagger\equiv\{E_i^*\}$, and $\phi$ is
trace-preserving if and only if $\phi^\dagger$ is unital. Strictly speaking, a CPTP map is defined on the set of trace-class operators on $\cl{H}$ and its dual map is defined on $\cl{L}(\cl{H})$. However, in the finite-dimensional case this set may be identified with $\cl{L}(\cl{H})$. We shall make use of this identification throughout the paper as it streamlines our presentation. Moreover, when it is important that the map under consideration is CPTP, capital script letters such as $\cl{E}$ and $\cl{F}$ will be used to denote it; otherwise, the Greek letter $\phi$ will be used.

For our purposes, an \emph{(operator) algebra} $\fA$ or $\fB$ will refer to a concrete finite-dimensional $C^{*}$-algebra \cite{Dav96};
that is, a set of operators $\fA$ inside $\mathcal{L}(\mathcal{H})$ for which there is an orthogonal direct
sum decomposition of the Hilbert space $\mathcal{H} = \oplus_k
(\mathcal{A}_k \otimes \mathcal{B}_k) \oplus \mathcal{K}$ such
that the algebra $\fA$ consists of all operators belonging to the
set
\begin{equation}\label{opalgform}
\fA = \oplus_k \big(I_{\mathcal{A}_k} \otimes
\mathcal{L}(\mathcal{B}_k)\big) \, \oplus \, 0_\mathcal{K},
\end{equation}
where $0_\mathcal{K}$ is the zero operator on $\mathcal{K}$ and $I_{\cl{A}_k}$ is the identity operator on $\cl{A}_k$. When it is known that the algebra's structure contains only a single summand $I_\cl{A} \otimes \mathcal{L}(\cl{B})$, we will denote it as $\fA_\mathcal{B} := I_\cl{A} \otimes \mathcal{L}(\cl{B})$.

%%%%%%%%%%%%%%%%%%%%%%%%%%%%%%%%%%%%%%%%%%%%%%%%%%%%%%%%%%%
\subsection{Representations}
%%%%%%%%%%%%%%%%%%%%%%%%%%%%%%%%%%%%%%%%%%%%%%%%%%%%%%%%%%%

Suppose $\fA$ is an algebra on $\cl{H}$. By a
\emph{representation} (or a \emph{$\ast$-homomorphism}) of $\fA$, we
mean a linear map $\pi: \fA \rightarrow \mathcal{L}(\mathcal{H})$ such that:
\begin{eqnarray*}
	\pi(ab) &=& \pi(a)\pi(b) \quad\quad \forall a,b\in\fA \\
	\pi(a^*) &=& \pi(a)^* \quad\quad\quad  \forall a\in\fA.
\end{eqnarray*}

\noindent Throughout this paper we will deal only with \emph{faithful} representations (that is, representations $\pi$ with ${\rm ker}(\pi) = \{0\}$). For every such representation $\pi$ of $\fA_\cl{B} = I_\cl{A}\otimes
\mathcal{L}(\mathcal{B})$, there is \cite{Dav96}
a positive integer $m$ and a unitary $U$ from
$\mathcal{B}^{\otimes m}$ into the range Hilbert space for $\pi$
such that
\begin{eqnarray}\label{repnform}
\pi(I_\cl{A} \otimes X) = U (I_m\otimes X) U^* \quad\quad \forall X\in\cl{L}(\cl{B}).
\end{eqnarray}

%%%%%%%%%%%%%%%%%%%%%%%%%%%%%%%%%%%%%%%%%%%%%%%%%%%%%%%%%%%
\subsection{Quantum Error Correction}\label{sec:qecprelim}
%%%%%%%%%%%%%%%%%%%%%%%%%%%%%%%%%%%%%%%%%%%%%%%%%%%%%%%%%%%

Standard quantum error correction considers quantum codes as
subspaces $\mathcal{C}\subseteq\mathcal{H}$
\cite{bennett96,Got96,KL95,shor95,steane96}. The code $\mathcal{C}$ is
said to be \emph{correctable} for $\mathcal{E}$ if there is a
CPTP map $\mathcal{R} :
\mathcal{L}(\mathcal{H})\rightarrow\mathcal{L}(\mathcal{H})$
such that $\mathcal{R}\circ \mathcal{E} \circ
\mathcal{P}_\mathcal{C} = \mathcal{P}_\mathcal{C}$, where
$\mathcal{P}_\mathcal{C}(\rho)=P_\mathcal{C} \rho P_\mathcal{C}$ and $P_{\cl{C}}$ is the orthogonal projection onto $\cl{C}$.

A generalization called ``operator quantum error correction''
\cite{KLP05,KLPL06} leads to the notion of \emph{subsystem codes}
\cite{AlyKla07,Bac05,KlaSar06,Pou05,ShaLid05}. Two Hilbert spaces
$\mathcal{A}$, $\mathcal{B}$ are \emph{subsystems} of
$\mathcal{H}$ when $\mathcal{H}$ decomposes as $\mathcal{H} =
\mathcal{C} \oplus \mathcal{C}^\perp$ with $\mathcal{C} =
\mathcal{A} \otimes \mathcal{B}$. A subsystem $\mathcal{B}$ is \emph{correctable} for
$\mathcal{E}$ if there is a CPTP map $\mathcal{R} :
\mathcal{L}(\mathcal{H})\rightarrow\mathcal{L}(\mathcal{H})$
and a CPTP map $\cl{F}_\cl{A} : \cl{L}(\cl{A})\rightarrow\cl{L}(\cl{A})$ such that $\mathcal{R}\circ \mathcal{E} \circ
\mathcal{P}_\mathcal{C} = (\mathcal{F}_\mathcal{A} \otimes {\rm
id}_\cl{B})\circ \mathcal{P}_\mathcal{C}$, where ${\rm id}_\cl{B}$ is the identity map on $\cl{B}$. As a notational convenience, given
operators $X\in\mathcal{L}(\mathcal{A})$ and
$Y\in\mathcal{L}(\mathcal{B})$, we will write $X\otimes Y$ for the
operator on $\mathcal{H}$ given by $(X\otimes Y) \oplus
0_{\mathcal{C}^\perp}$.

The following simple example is provided to help illustrate these ideas.

\begin{exam}\label{exam:ec01}
{\rm Consider the $2 \times 2$ identity operator and Pauli bit flip operator $X$ represented in the standard basis $\{ \ket{0}, \ket{1} \}$ as follows:
\begin{align*}
	I := \begin{bmatrix}1 & 0 \\ 0 & 1\end{bmatrix} \quad X := \begin{bmatrix}0 & 1 \\ 1 & 0\end{bmatrix}.
\end{align*}

Let $\cl{H}$ be a three-qubit ($8$-dimensional) Hilbert space with basis vectors $\ket{ijk} = \ket{i} \otimes \ket{j} \otimes \ket{k}$ and consider the CPTP map $\cl{E}$ given by the four Kraus operators
\begin{align*}
 \frac{1}{2} I \otimes I \otimes I, \ \ \frac{1}{2} X \otimes I \otimes I, \ \ \frac{1}{2} I \otimes X \otimes I, \ \ \frac{1}{2} I \otimes I \otimes X.
\end{align*}

\noindent Physically, this map can be interpreted as having an equal probability of not applying an error, applying a bit flip on the first qubit, applying a bit flip on the second qubit, or applying a bit flip on the third qubit.

It is not difficult to check that $\cl{C}_0 := {\rm span}\{\ket{000},\ket{111}\}$ is a correctable subspace for this map. Indeed, if we define three more subspaces
\begin{align*}
	\cl{C}_1 & := {\rm span}\{\ket{100},\ket{011}\}, \\
	\cl{C}_2 & := {\rm span}\{\ket{010},\ket{101}\}, \\
	\cl{C}_3 & := {\rm span}\{\ket{001},\ket{110}\},
\end{align*}

\noindent then some simple algebra verifies that one valid correction operation for $\cl{E}$ on $\cl{C}_0$ is defined by the following four Kraus operators:
\begin{align*}
	\cl{R} \equiv \Big\{ P_{\cl{C}_0}, \ \ (X \otimes I \otimes I) P_{\cl{C}_1}, \ \ (I \otimes X \otimes I) P_{\cl{C}_2}, \ \ (I \otimes I \otimes X) P_{\cl{C}_3} \Big\}.
\end{align*}}
\end{exam}

%%%%%%%%%%%%%%%%%%%%%%%%%%%%%%%%%%%%%%%%%%%%%%%%%%%%%%%%%%%
\section{Generalized Multiplicative Domains}\label{sec:multDom}
%%%%%%%%%%%%%%%%%%%%%%%%%%%%%%%%%%%%%%%%%%%%%%%%%%%%%%%%%%%

The {\em multiplicative domain} \cite{Cho74,Paulsentext} of a CP map $\phi : \cl{L}(\cl{H}) \rightarrow \cl{L}(\cl{H})$, denoted $MD(\phi)$, is defined as
\begin{align}\label{eq:multdom_def}\begin{split}
	MD(\phi) := \{ a \in \cl{L}(\cl{H}) : \phi(ab) = \phi(a)\phi(b) \text{ and } \phi(ba) = \phi(b)\phi(a) \ \forall \, b \in \cl{L}(\cl{H}) \}.
\end{split}\end{align}

It is clear that $MD(\phi)$ is a $C^{*}$-algebra, and hence has a structure as in Eq.~(\ref{opalgform}). The question of what role the multiplicative domain plays in quantum error correction was answered in \cite{CJK08}. In particular, it was shown that if a CPTP map $\cl{E}$ is unital, then its multiplicative domain encodes exactly its unitarily correctable codes.

Similarly, given a $C^*$-subalgebra $\fA$ of $\cl{L}(\cl{H})$, we say that $\phi$ is an {\em $\fA$-bimodule map} (or that $\fA$ is a $\phi$-bimodule) \cite{Paulsentext} if
\begin{align}\label{eq:bimodule_def}\begin{split}
	\phi(ab) = a\phi(b) \text{ and } \phi(ba) = \phi(b)a \quad \forall \, a \in \fA, b \in \cl{L}(\cl{H}).
\end{split}\end{align}

We will see later that if $\cl{E}$ is a CPTP map, then $\cl{E}$-bimodules are exactly the noiseless subsystems for $\cl{E}$. Here we will generalize the multiplicative domain and bimodules to a wider class of algebras that we will later see also have a connection with quantum error correction.

\begin{defn}\label{defn:genmultdom}Given a subalgebra $\fA$ of $\cl{L}(\cl{H})$, a representation $\pi : \fA \rightarrow \cl{L}(\cl{H})$, and a CP map $\phi : \cl{L}(\cl{H}) \rightarrow \cl{L}(\cl{H})$, we define the \emph{multiplicative domain of $\phi$ with respect to $\pi$} as
	\[
	MD_{\pi}(\phi) :=  \big\{ a \in \fA : \pi(a)\phi(b) = \phi(ab) \text{ and } \phi(b)\pi(a) = \phi(ba) \ \ \forall \, b \in \fA \big\}.
	\]
\end{defn}

The sets $MD_\pi(\phi)$ are clearly $C^{*}$-algebras. Furthermore, they generalize the standard multiplicative domain, as $\phi$ acts as a representation when restricted to $MD(\phi)$, allowing us to set $\pi := \phi|_{MD(\phi)}$ to achieve $MD_\pi(\phi) = MD(\phi)$. Similarly, they generalize $\phi$-bimodules because one can take $\pi := {\rm id}_\fA$, the identity representation on some bimodule $\fA$, to achieve $MD_\pi(\phi) = \fA$.

%%%%%%%%%%%%%%%%%%%%%%%%%%%%%%%%%%%%%%%%%%%%%%%%%%%%%%%%%%%
\subsection{Characterizing Generalized Multiplicative Domains}
%%%%%%%%%%%%%%%%%%%%%%%%%%%%%%%%%%%%%%%%%%%%%%%%%%%%%%%%%%%

Our first question is whether or not any of the well-known results on the standard multiplicative domain extend to this more general setting. In particular, we recall the following result \cite{Cho74,Paulsentext} that shows how the multiplicative domain simplifies in the unital case.

\begin{thm}\label{thm:multDomChoi}
  Let $\phi : \cl{L}(\cl{H}) \rightarrow \cl{L}(\cl{H})$ be a completely positive, unital map. Then
  \begin{align*}
   MD(\phi) = & \big\{ a \in \cl{L}(\cl{H}) : \phi(a)^{*}\phi(a) =
\phi(a^*a)\text{ and } \phi(a)\phi(a)^{*} =
\phi(aa^*)\big\}.
  \end{align*}
\end{thm}

A natural question to ask of the generalized multiplicative domains is whether or not they have an internal characterization analogous to that of Theorem~\ref{thm:multDomChoi}. In particular, if $\phi$ is a unital map and $a \in \fA$ is such that $\phi(a)^*\pi(a) = \phi(a^*a)$ and $\pi(a)\phi(a)^* = \phi(aa^*)$, does it follow that $a \in MD_\pi(\phi)$? We begin with a brief example to show that the answer to this question is ``no'', even if we make the additional restriction that $\phi$ must be trace-preserving.

\begin{exam}\label{counterexample}
{\rm 	Let $\cl{E} : M_4 \rightarrow M_4$ be the CPTP map defined by the following two Kraus operators in the standard basis $\{ \ket{00}, \ket{01}, \ket{10}, \ket{11} \}$:
	\begin{align*}
		E_1 := \frac{1}{2}\begin{bmatrix}1 & 0 & 1 & 0 \\ 0 & 1 & 1 & 0 \\ 1 & 1 & 0 & 0 \\ 0 & 0 & 0 & \sqrt{2}\end{bmatrix} \quad \quad E_2 := \frac{1}{2}\begin{bmatrix}1 & 0 & -1 & 0 \\ 0 & -1 & 1 & 0 \\ -1 & 1 & 0 & 0 \\ 0 & 0 & 0 & \sqrt{2}\end{bmatrix}.
	\end{align*}
	
	It is not difficult to verify that this map is unital and trace-preserving. Now define $a := \ketbra{00}{00}$ to be the rank-$1$ projection onto the first standard basis vector. Let $\pi$ be the representation of the operators $A$ supported on ${\rm span}\{\ket{00}, \ket{01}\}$ defined as follows:
	\begin{align*}
		\pi\Big( \begin{bmatrix}A & 0 \\ 0 & 0 \end{bmatrix} \Big) := \begin{bmatrix}A & 0 \\ 0 & A \end{bmatrix}.
	\end{align*}
	
	Simple algebra reveals that
	\begin{align*}
		\cl{E}\Bigg( \begin{bmatrix} c_1 & c_2 & 0 & 0 \\ c_3 & c_4 & 0 & 0 \\ 0 & 0 & 0 & 0 \\ 0 & 0 & 0 & 0 \end{bmatrix} \Bigg) = \frac{1}{2}\begin{bmatrix} c_1 & 0 & c_2 & 0 \\ 0 & c_4 & c_3 & 0 \\ c_3 & c_2 & c_1 + c_4 & 0 \\ 0 & 0 & 0 & 0 \end{bmatrix}.
	\end{align*}
	
\noindent It is then not difficult to verify that $\cl{E}(a^2) = \cl{E}(a)\pi(a) = \pi(a)\cl{E}(a)$. However, $\cl{E}(ba) = \cl{E}(b)\pi(a)$ and $\cl{E}(ab) = \pi(a)\cl{E}(b)$ if and only if the support of $b$ is contained in the support of $a$. In particular, these equations do not hold for all $b$ supported on ${\rm span}\{\ket{00}, \ket{01}\}$.}
\end{exam}

In spite of such counterexamples, there is indeed an internal characterization of $MD_\pi(\phi)$ along the lines of what we are looking for, though it is slightly more delicate than the case of the multiplicative domain. Interestingly, we require the map to be not only unital but trace-preserving as well. We also require that $a$ be positive and have full rank inside of $\fA$. The above example fails because the range of $\fA$ has dimension $2$, but $a$ has rank $1$. As a notational convenience, we shall write $\fA_{>0}$ for the set of positive operators inside $\fA$ of maximal rank.

\begin{thm}\label{thm:genChar} If $\cl{E} : \cl{L}(\cl{H}) \rightarrow \cl{L}(\cl{H})$ is a unital CPTP map and $\pi : \fA \rightarrow \cl{L}(\cl{H})$ is a representation with $\fA$ a subalgebra of $\cl{L}(\cl{H})$, then
\[
	MD_{\pi}(\cl{E}) = {\rm span}\Big\{ a \in \fA_{>0} : \cl{E}(a)\pi(a) = \cl{E}(a^2) = \pi(a)\cl{E}(a) \Big\}.
\]
\end{thm}

\begin{proof}
	The algebra $MD_{\pi}(\cl{E})$ is clearly contained in this span. Without loss of generality we may assume that $\fA = 1_n\otimes
\mathcal{L}(\cl{B})$ and $a = 1_n \otimes a_\cl{B}$. Then by Equation~\eqref{repnform} we know that there exists a unitary $U$ such that $\pi(1_n \otimes a_\cl{B}) \equiv U(1_m \otimes a_\cl{B})U^*$. Then conjugating the given equations by $U^*$ reveals that
\begin{align}\label{eq:genCharWLOG}
  \cl{U} \circ \cl{E}(a)(1_m \otimes a_\cl{B}) = \cl{U} \circ \cl{E}(a^2) = (1_m \otimes a_\cl{B})\cl{U} \circ \cl{E}(a),
\end{align}

\noindent where $\cl{U}(X) \equiv U^*XU$. We can thus assume without loss of generality that $\pi(1_n \otimes a_\cl{B}) = 1_m \otimes a_\cl{B}$ and in particular that $a$ commutes with $\pi(a)$.

  Now let $a \in \fA_{>0}$ be such that $\cl{E}(a)\pi(a) = \cl{E}(a^2) = \pi(a)\cl{E}(a)$. Use Stinespring's Dilation Theorem to write $\cl{E}(a) \equiv V^*(id_k \otimes a)V$ for some isometry $V : \cl{H} \rightarrow \cl{H}^{\otimes k}$.
	
	Let $\ket{x}$ be an arbitrary eigenvector of $\pi(a)$ with corresponding eigenvalue $\lambda$. Then we have that
	\begin{align}\label{eq:genChar01}\begin{split}
		0 & = \cl{E}(a^2)\ket{x} - \cl{E}(a)\pi(a)\ket{x} \\
		& = \cl{E}(a(a - \lambda 1_{\fA}))\ket{x} \\
		& = V^*(id_k \otimes a(a - \lambda 1_{\fA}))V\ket{x},
	\end{split}\end{align}
	
	\noindent where $1_{\fA}$ is the unit element of $\fA$. Now, if $V\ket{x}$ were in the nullspace of $id_k \otimes a(a - \lambda 1_{\fA})$ then we could multiply on the left by $V^*(id_k \otimes ba^{-1})$, where $a^{-1}$ is the inverse of $a$ inside $\fA$ and $b \in \fA$ is arbitrary, giving us
\begin{align*}
	\cl{E}(ba)\ket{x} & = \cl{E}(b)\pi(a)\ket{x}.
\end{align*}
	
\noindent Since $a \geq 0$, there exists an orthonormal basis of eigenvectors for $\pi(a)$. Because $\ket{x}$ was an arbitrary eigenvector, it would follow that $\cl{E}(ba) = \cl{E}(b)\pi(a)$ for all $b \in \fA$. It is thus enough to show that
\[
  (id_k \otimes a(a - \lambda 1_{\fA}))V\ket{x} = 0
\]

\noindent for all eigenvalue/eigenvector pairs $(\lambda, \ket{x})$ of $\pi(a)$.
	
	From now on, it will be convenient to write $V^*$ as a row operator $V^* = [ E_1 \ E_2 \ \ldots \ E_k ]$, where $\cl{E} \equiv \big\{ E_j \big\}$. Note that trace-preservation of $\cl{E}$ implies
	\begin{align}\label{eq:TPcond}
		\sum_i \| E_i\ket{x_j} \|^2 = \sum_i\bra{x_j}E_i^*E_i\ket{x_j} = 1	 \quad \forall \, j.
	\end{align}
	
\noindent Similarly, $\cl{E}$ being unital implies
	\begin{align}\label{eq:Unitalcond}
		\sum_i \| E_i^*\ket{x_j} \|^2 = \sum_i\bra{x_j}E_iE_i^*\ket{x_j} = 1	 \quad \forall \, j.
	\end{align}
	
	For now assume that we are dealing with eigenvalue/eigenvector pairs $(\lambda,\ket{x_1}), \ldots, (\lambda,\ket{x_{m}})$ that correspond to the minimal eigenvalue $\lambda$ of $\pi(a)$, which has multiplicity $m$. Let $\cl{V}_\lambda$ be the span of $\ket{x_j}, j = 1, 2, \ldots, m$. It then follows from Equation~\eqref{eq:genChar01} that $a(a - \lambda 1_{\fA})E_i^*\ket{x_j} = 0$ for all $j = 1, 2, \ldots, m$ and all $i$ because $a(a - \lambda 1_{\fA}) \geq 0$. This implies that, for all $j = 1, 2, \ldots, m$ and all $i$, $E_i^*\ket{x_j} \in \cl{V}_\lambda$ -- i.e., $\cl{V}_\lambda$ is invariant for each $E_i^*$.
	
	By using Equation~\eqref{eq:Unitalcond} $m$ times we see that we need $\sum_i \sum_{j=1}^{m} \| E_i^*\ket{x_j} \|^2 = m$. Using the facts that $\cl{V}_\lambda$ is invariant for each $E_i^*$ and ${\rm dim}(\cl{V}_\lambda) = m$, we see that this implies $\sum_i \sum_{j=1}^{m} \| E_i\ket{x_j} \|^2 \geq m$. However, because $\cl{E}$ is trace-preserving, this implies via Equation~\eqref{eq:TPcond} that $E_i\ket{x_j} \in \cl{V}_\lambda$ (i.e., $\cl{V}_\lambda$ is $E_i$-invariant) for each $i$. Hence $\cl{V}_\lambda$ is a reducing subspace for each $E_i$ and because $a$, $\pi(a)$ and $\cl{E}(a)$ commute, it then follows by the Spectral Theorem that the entire problem decomposes as a direct sum
	\begin{align*}
		\cl{E} & = \cl{E}_1 \oplus \cl{E}_2 \\
		\pi & = \pi_1 \oplus \pi_2 \\
		a & = a_1 \oplus a_2,
	\end{align*}

	\noindent where $\cl{E}_1, \pi_1 : \cl{V}_\lambda \rightarrow \cl{V}_\lambda$, $\cl{E}_2, \pi_2 : \cl{V}_\lambda^\perp \rightarrow \cl{V}_\lambda^\perp$, $a_1 \in \cl{V}_\lambda$, and $a_2 \in \cl{V}_\lambda^\perp$. It follows that
	\begin{align*}
		\cl{E}_2(a_2)\pi_2(a_2) = \cl{E}_2(a_2^2) = \pi_2(a_2)\cl{E}_2(a_2).
	\end{align*}
	
	Now we simply repeat this procedure with the smallest eigenvalue $\mu$ of $\pi_2(a_2)$ to decompose the problem into a direct sum again, and repeat this way until we have exhausted all of the eigenvalues of $\pi(a)$. It follows that $(id_k \otimes a(a - \lambda 1_{\fA}))V\ket{x} = 0$ for all eigenvalue/eigenvector pairs $(\lambda, \ket{x})$ of $\pi(a)$, and so it follows that $\cl{E}(ba) = \cl{E}(b)\pi(a)$ for all $b \in \fA$.
	
	The other required equality follows similarly by using $\pi(a)\cl{E}(a) = \cl{E}(a^2)$. The proof is completed by recalling that an algebra is spanned by its strictly positive elements.
\end{proof}

	Even though we have an internal characterization of the generalized multiplicative domains via Theorem~\ref{thm:genChar}, it would be nice to have one that did not rely on us observing only positive elements and taking linear combinations. In order to motivate our second characterization of generalized multiplicative domains, we present without proof a well-known result about bimodules \cite[Exercise 4.3 (ii)]{Paulsentext}.
	\begin{thm}\label{thm:bimodules}
		Let $\phi : \cl{L}(\cl{H}) \rightarrow \cl{L}(\cl{H})$ be completely positive and let $\fA$ be a subalgebra of $\cl{L}(\cl{H})$ containing the identity operator $1$ of $\cl{L}(\cl{H})$. Then $\phi$ is an $\fA$-bimodule map if and only if
		\begin{align*}
			\phi(1)a = \phi(a) = a\phi(1) \quad \forall \, a \in \fA.
		\end{align*}
	\end{thm}
	
	Because we have seen that the generalized multiplicative domains are really bimodules in the case $\pi = {\rm id_\fA}$, we might na\"{i}vely expect that replacing $a$ by $\pi(a)$ and $1$ by $1_\fA$ on the left and right hand sides in Theorem~\ref{thm:bimodules} will give us an analogous result for generalized multiplicative domains. This intuition leads to the following characterization theorem, which we will see in Section~\ref{sec:qec} has implications in quantum error correction.

\begin{thm}\label{thm:genChar2} Let $\cl{E} : \cl{L}(\cl{H}) \rightarrow \cl{L}(\cl{H})$ be a unital CPTP map and let $\pi : \fA \rightarrow \cl{L}(\cl{H})$ be a representation with $\fA$ a subalgebra of $\cl{L}(\cl{H})$. Then
\[
	MD_{\pi}(\cl{E}) = \Big\{ a \in \fA : \cl{E}(1_\fA)\pi(a) = \cl{E}(a) = \pi(a)\cl{E}(1_\fA) \Big\}.
\]
\end{thm}
	
	In order to prove Theorem~\ref{thm:genChar2}, we first must prove a pair of lemmas.
	
\begin{lemma}\label{lem:genChar02}
	Suppose $\phi : \cl{L}(\cl{H}) \rightarrow \cl{L}(\cl{H})$ is a CP map, $\pi : \fA \rightarrow \cl{L}(\cl{H})$ is a representation, and $a \in \fA$ is such that $\phi(1_\fA)\pi(a) = \phi(a) = \pi(a)\phi(1_\fA)$. If $b, c \in \fA$ are Hermitian such that $a = b + ic$ then $\phi(1_\fA)\pi(b) = \phi(b) = \pi(b)\phi(1_\fA)$ and $\phi(1_\fA)\pi(c) = \phi(c) = \pi(c)\phi(1_\fA)$.
\end{lemma}
\begin{proof}
	By hypothesis, we have that $\phi(1_\fA)\pi(a) = \pi(a)\phi(1_\fA)$, and hence
	\begin{align*}
		i(\phi(1_\fA)\pi(c) - \pi(c)\phi(1_\fA)) = \pi(b)\phi(1_\fA) - \phi(1_\fA)\pi(b).
	\end{align*}
	
\noindent Since the left hand side is Hermitian and the right hand side is skew-Hermitian, both sides must equal zero. Thus $\phi(1_\fA)\pi(b) = \pi(b)\phi(1_\fA)$ and $\phi(1_\fA)\pi(c) = \pi(c)\phi(1_\fA)$. Notice that this implies $\phi(1_\fA)\pi(b)$ and $\phi(1_\fA)\pi(c)$ are both Hermitian. By using the hypotheses of the lemma again, we see that $\phi(a) = \phi(1_\fA)\pi(a)$ implies
	\[
		\phi(b) + i\phi(c) = \phi(1_\fA)\pi(b) + i\phi(1_\fA)\pi(c).
	\]
	
\noindent As $\phi(1_\fA)\pi(b)$ and $\phi(1_\fA)\pi(c)$ are Hermitian, it follows that $\phi(b) = \phi(1_\fA)\pi(b)$ and $\phi(c) = \phi(1_\fA)\pi(c)$, completing the proof.
\end{proof}

\begin{lemma}\label{lem:genChar03}
	If $\phi : \cl{L}(\cl{H}) \rightarrow \cl{L}(\cl{H})$ is a CP map and $\pi : \fA \rightarrow \cl{L}(\cl{H})$ is a representation then the set
	$$
	  \Big\{ a \in \fA : \phi(1_\fA)\pi(a) = \phi(a) = \pi(a)\phi(1_\fA) \Big\}
	$$
	\noindent is spanned by its positive elements.
\end{lemma}
\begin{proof}
	It follows from Lemma~\ref{lem:genChar02} that the given set is spanned by its Hermitian elements. To see that it is spanned by its positive elements, simply note that if $\phi(1_\fA)\pi(a) = \phi(a)$ then, for any $r \in \bb{R}$, we trivially have $\phi(1_\fA)\pi(a + r 1_\fA) = \phi(a + r 1_\fA)$ (and similar for the other equality).
\end{proof}

\begin{proof}[Proof of Theorem~\ref{thm:genChar2}]
	The ``$\subseteq$'' direction of the proof is trivial.
	
	For the ``$\supseteq$'' direction of the proof, let $a \in \fA$ be such that $\cl{E}(1_\fA)\pi(a) = \cl{E}(a) = \pi(a)\cl{E}(1_\fA)$. By Lemma~\ref{lem:genChar03} we can assume that $a \geq 0$. The rest of the proof proceeds almost identically to the proof of Theorem~\ref{thm:genChar}.
\end{proof}

One question that naturally arises at this point is whether or not the map $\cl{E}$ really needs to be trace-preserving in order for the internal characterizations of Theorems~\ref{thm:genChar} and \ref{thm:genChar2} to work; after all, for the multiplicative domain it was enough for $\cl{E}$ to just be unital. The following example shows that both of these characterizations fail if we remove either of the trace-preservation or unital hypotheses.

\begin{exam}\label{exam:needTP}
	{\rm Let $\phi : M_3 \rightarrow M_3$ be the CP map defined by the following two Kraus operators in the standard basis $\{ \ket{0}, \ket{1}, \ket{2}\}$:
	\begin{align*}
		E_1 := \frac{1}{2}\begin{bmatrix}\sqrt{2} & 0 & 0 \\ 1 & 0 & 1 \\ 0 & 0 & \sqrt{2}\end{bmatrix} \quad \quad E_2 := \frac{1}{2}\begin{bmatrix}\sqrt{2} & 0 & 0 \\ -1 & 0 & -1 \\ 0 & 0 & \sqrt{2} \end{bmatrix}.
	\end{align*}
	
\noindent It is not difficult to verify that this map is unital but not trace-preserving.
	
	Now define $a := \ketbra{0}{0} + \frac{5}{2}\ketbra{1}{1} + 3\ketbra{2}{2}$. Let $\pi$ simply be the identity map. Then some algebra reveals that
	\begin{align*}
		\pi(a)\phi(a) = \phi(a^2) = \phi(a)\pi(a).
	\end{align*}
	
\noindent However, taking $b = \ketbra{0}{0}$ gives us that
	\begin{align*}
		\phi(ba) = \ketbra{0}{0} + \frac{1}{2}\ketbra{1}{1} \neq \phi(b)\pi(a) = \ketbra{0}{0} + \frac{5}{4}\ketbra{1}{1}.
	\end{align*}
	
	Similarly, to consider the internal characterization of Theorem~\ref{thm:genChar2}, consider the same map $\phi$ but set $a := \ketbra{0}{0} + 2\ketbra{1}{1} + 3\ketbra{2}{2}$. Let $\pi$ again be the identity map. Then
	\begin{align*}
		\pi(a)\phi(I) = \phi(a) = \phi(I)\pi(a) = \ketbra{0}{0} + 2\ketbra{1}{1} + 3\ketbra{2}{2}.
	\end{align*}
	
\noindent However, again taking $b = \ketbra{0}{0}$ gives us that
	\begin{align*}
		\phi(ba) = \ketbra{0}{0} + \frac{1}{2}\ketbra{1}{1} \neq \phi(b)\pi(a) = \ketbra{0}{0} + \ketbra{1}{1}.
	\end{align*}
	
	To instead see that trace-preservation of $\phi$ without it being unital is not sufficient for Theorems~\ref{thm:genChar} and~\ref{thm:genChar2}, consider the same operators $a$ and $b$ and let $\pi$ be the identity map as before, but instead of the unital map $\phi$ use the CPTP map $\cl{E}$ defined by the following three Kraus operators:
	\begin{align*}
		E_1 := \frac{1}{\sqrt{2}}\begin{bmatrix}1 & 0 & 0 \\ 0 & \sqrt{2} & 0 \\ 0 & 0 & 1\end{bmatrix} \quad E_2 := \frac{1}{\sqrt{2}}\begin{bmatrix}0 & 0 & 0 \\ 1 & 0 & 0 \\ 0 & 0 & 0 \end{bmatrix} \quad E_3 := \frac{1}{\sqrt{2}}\begin{bmatrix}0 & 0 & 0 \\ 0 & 0 & 1 \\ 0 & 0 & 0 \end{bmatrix}.
	\end{align*}}
\end{exam}

%%%%%%%%%%%%%%%%%%%%%%%%%%%%%%%%%%%%%%%%%%%%%%%%%%%%%%%%%%%
\section{Representations in Quantum Error Correction}\label{sec:qec}
%%%%%%%%%%%%%%%%%%%%%%%%%%%%%%%%%%%%%%%%%%%%%%%%%%%%%%%%%%%

Theorem~\ref{thm:qec_main} below gave the initial motivation from quantum information for the investigation of the generalized multiplicative domains that have been introduced. The equivalence of conditions (1) and (2) was proved in \cite{CJK08}; condition (2) is presented here to provide a more complete picture, as it is a seemingly weaker condition than condition (3) in that (3) trivially implies (2). Additionally, condition (2) was our original motivation for investigating the internal characterization of generalized multiplicative domains that was presented in Theorem~\ref{thm:genChar2}.

The equivalence of (1) and (3) is in the same flavour as the commutant characterization of noiseless subsystems for unital CPTP maps given in \cite{HKL04}, where by ``noiseless'' we mean a correctable subsystem that is corrected by the map $\mathcal{R} \equiv {\rm id}$. In fact, by taking $\pi = {\rm id}$, it can be thought of as a generalization of Theorem~1 of \cite{CK06}, which states that if $\cl{C} = \cl{A} \otimes \cl{B}$ is a subspace of $\cl{H}$ then $\cl{B}$ is noiseless for $\cl{E}$ if and only if $aE_iP_\cl{C} - E_ia = P_\cl{C}E_i^*a - aE_i^* = 0$ for all $a \in \fA_\cl{B}$ and all $i$, where $\cl{E} \equiv \{ E_i \}$. This agrees with our representation-theoretic picture of error correction, as in \cite{CJK08} it was noted that $\pi^\dagger$ acts as a recovery operation when restricted to $\cl{E}(\fA_\cl{B})$, so the recovery operation is just the identity in the noiseless subsystem case, as it should be. This also shows that $\cl{E}$-bimodules are exactly noiseless subsystems, as we saw earlier that $\pi = {\rm id}$ corresponds to the case when generalized multiplicative domains are bimodules.

The equivalence of (1) and (4) is our primary motivation for investigating generalized multiplicative domains. In much the same way that the standard multiplicative domain encodes the unitarily correctable codes of unital CPTP maps, we see that the generalized multiplicative domains encode {\em all} correctable codes for {\em arbitrary} CPTP maps.

\begin{thm}\label{thm:qec_main}
	Let $\cl{E} \equiv \{E_i\} : \cl{L}(\cl{H}) \rightarrow \cl{L}(\cl{H})$ be a CPTP map and let $\cl{C} = \cl{A} \otimes \cl{B}$ be a subspace of $\cl{H}$. Then the following are equivalent:
 	\begin{enumerate}
		\item $\cl{B}$ is a correctable subsystem for $\cl{E}$,
		\item $\exists$ a representation $\pi : \fA_\cl{B} \rightarrow \cl{L}(\cl{H})$ such that
		\[
		  \pi(a)\cl{E}(P_\cl{C}) = \cl{E}(P_\cl{C})\pi(a) = \cl{E}(a) \quad \forall \, a \in \fA_\cl{B},
		\]
		\item $\exists$ a representation $\pi : \fA_\cl{B} \rightarrow \cl{L}(\cl{H})$ such that
		\[
		  \pi(a)E_iP_\cl{C} - E_ia = P_\cl{C}E_i^*\pi(a) - aE_i^* = 0 \quad \forall \, a \in \fA_\cl{B}, \forall \, i,
		\]
		\item $\exists$ a representation $\pi : \fA_\cl{B} \rightarrow \cl{L}(\cl{H})$ such that
		\[ MD_\pi(\cl{E}) = \fA_\cl{B}. \]
	\end{enumerate}
\end{thm}
\begin{proof}
	We prove the theorem via the implications (1) $\Rightarrow$ (3) $\Rightarrow$ (4) $\Rightarrow$ (2) $\Rightarrow$ (1).
	
	To prove that (1) $\Rightarrow$ (3), recall from Theorem~6 of \cite{CJK08} that there exist unitaries $\big\{ U_i \big\}$ and operators $N_{i,j}$ such that $P_\cl{C}U_i^*U_jP_\cl{C} = \delta_{ij}P_\cl{C}$ and
	\[
	  \cl{E}(b_\cl{A} \otimes b_\cl{B}) \equiv \sum_{i,j}U_i(N_{i,j}b_\cl{A}N_{i,j}^* \otimes b_\cl{B})U_i^*.
	\]
\noindent Let $a = (I_\cl{A} \otimes a_\cl{B}) \in \fA_\cl{B}$. Defining $\pi(a) := \sum_l U_l a U_l^* $ we see that
	\begin{align*}
		\pi(a)U_i(N_{i,j} \otimes I_\cl{B})P_\cl{C} & = \sum_l U_l (I_\cl{A} \otimes a_\cl{B}) U_l^* U_i(N_{i,j} \otimes I_\cl{B})P_\cl{C} \\
		& = U_i (I_\cl{A} \otimes a_\cl{B}) (N_{i,j} \otimes I_\cl{B})P_\cl{C} \\
		& = U_i (N_{i,j} \otimes I_\cl{B})P_\cl{C} (I_\cl{A} \otimes a_\cl{B}).		
	\end{align*}
	
\noindent It is well-known that two sets of Kraus operators define the same CP map if and only if their linear spans are equal. It follows that there exist constants $c_{i,j,k}$ such that if $\cl{E} \equiv \big\{ E_k \big\}$, then
	\begin{align*}
		E_k P_\cl{C} = \sum_{i,j}c_{i,j,k}U_i(N_{i,j} \otimes I_\cl{B})P_\cl{C}.
	\end{align*}
	
\noindent It then follows immediately that
	\begin{align}\label{eq:qec_main_01}
		\pi(a)E_k P_\cl{C} = E_k a.
	\end{align}
	
\noindent The other equality is proved in a similar manner. \\
	
	To see that (3) $\Rightarrow$ (4), simply multiply Equation~\eqref{eq:qec_main_01} on the right by $bE_k^*$, where $b \in \fA_{\cl{B}}$ is arbitrary. Then summing over all $k$ gives us for all $b \in \fA_\cl{B}$,
	\begin{align*}
		\pi(a)\cl{E}(b) = \pi(a)\sum_k E_kbE_k^* = \sum_k E_kabE_k^* = \cl{E}(ab).
	\end{align*}
	
\noindent The proof of the other necessary equality is obvious. \\
	
	To see that (4) $\Rightarrow$ (2), simply take $b = P_\cl{C}$ in $MD_\pi(\cl{E})$. \\

	The following proof that (2) $\Rightarrow$ (1) was originally presented in \cite{CJK08}, but is provided here for completeness.
	
To see (2) $\Rightarrow$ (1), we show that the algebra
$\fA_\mathcal{B}$ may be precisely corrected, which is equivalent
to correcting the subsystem $\mathcal{B}$ (see Theorem~3.2 of
\cite{KLPL06} for instance). First note that the representation
$\pi$ defines a subspace and subsystems $\mathcal{C}' =
\mathcal{A}' \otimes \mathcal{B}'$ with $\mathcal{B}'$ the same
dimension as $\mathcal{B}$ and an isometry $V: \mathcal{B}
\rightarrow \mathcal{B}'$ such that
\[
\pi(I_\mathcal{A}\otimes \rho_\mathcal{B}) = I_{\mathcal{A}'}
\otimes \mathcal{V} (\rho_\mathcal{B}) \quad \forall
\rho_\mathcal{B},
\]
where $\mathcal{V} (\rho_\mathcal{B}) = V \rho_\mathcal{B}
V^*$. Further, as $\mathcal{E} (P_\mathcal{C})$ commutes
with $\pi(\fA_\mathcal{B})$, it follows that $P_{\mathcal{C}'}
\mathcal{E} (P_\mathcal{C}) P_{\mathcal{C}'} =
\sigma_{\mathcal{A}'} \otimes I_{\mathcal{B}'}$ for some positive
operator $\sigma_{\mathcal{A}'}\in\mathcal{L}(\mathcal{A}')$ with
trace equal to $\dim\mathcal{C}$. Thus we have for all
$\rho_\mathcal{B}$,
\begin{eqnarray*}
\mathcal{E}(I_\mathcal{A}\otimes \rho_\mathcal{B}) &=&
\pi(I_\mathcal{A}\otimes
\rho_\mathcal{B})\mathcal{E}(P_\mathcal{C}) \\ &=&
(I_{\mathcal{A}'}\otimes
\mathcal{V}(\rho_\mathcal{B}))(\sigma_{\mathcal{A}'}\otimes
I_{\mathcal{B}'}) \\ &=& \sigma_{\mathcal{A}'} \otimes
\mathcal{V}(\rho_\mathcal{B}).
\end{eqnarray*}
Now define a CPTP map $\mathcal{R}$ on $\mathcal{H}$ such that
$\mathcal{R} \circ \mathcal{P}_{\mathcal{C}'} =
(\mathcal{D}_{\mathcal{A}|\mathcal{A}'} \otimes
\mathcal{V}^\dagger) \circ \mathcal{P}_{\mathcal{C}'}$, where
$\mathcal{D}_{\mathcal{A}|\mathcal{A}'}$ is the completely
depolarizing map from $\mathcal{A}'$ to $\mathcal{A}$, and it
follows that
$(\mathcal{R}\circ\mathcal{E})(I_\mathcal{A}\otimes\rho_\mathcal{B})
= I_\mathcal{A} \otimes \rho_\mathcal{B}$ for all
$\rho_\mathcal{B}$. This shows $\fA_\mathcal{B}$ can be exactly
corrected, and completes the proof.
\end{proof}

It is worth noting that we can actually use this error correction result to obtain another internal characterization of generalized multiplicative domains. Indeed, much like we saw an internal characterization in the same flavour as condition (2) of Theorem~\ref{thm:qec_main}, we now present an internal characterization based on condition (3). Note that, unlike the previous internal characterizations, this result holds for arbitrary (i.e., not necessarily unital) CPTP maps.

\begin{cor}\label{cor:genChar2} Let $\cl{E} : \cl{L}(\cl{H}) \rightarrow \cl{L}(\cl{H})$ be a CPTP map and let $\pi : \fA \rightarrow \cl{L}(\cl{H})$ be a representation with $\fA$ a subalgebra of $\cl{L}(\cl{H})$. Then
\[
	MD_{\pi}(\cl{E}) = \Big\{ a \in \fA : \pi(a)E_iP_\cl{C} - E_ia = P_\cl{C}E_i^*\pi(a) - aE_i^* = 0 \quad \forall \, i \Big\}.
\]
\end{cor}
\begin{proof}
	For the ``$\subseteq$'' direction of the proof, first recall that $MD_{\pi}(\cl{E})$ has a structure as in Eq.~(\ref{opalgform}). Then, restricting $\pi$ to one particular summand $\fA_\cl{B}$ gives $MD_\pi(\cl{E}) = \fA_\cl{B}$. Theorem~\ref{thm:qec_main} then implies the required equalities.
	
	The ``$\supseteq$'' direction of the proof comes from multiplying $\pi(a)E_iP_\cl{C} - E_ia = 0$ on the right by $bE_i^*$ and $P_\cl{C}E_i^*\pi(a) - aE_i^* = 0$ on the left by $E_i b$, and then summing over $i$ to obtain
	\begin{align*}
		\pi(a)\cl{E}(b) = \cl{E}(ab) \ \text{ and } \ \cl{E}(b)\pi(a) = \cl{E}(ba).
	\end{align*}
\end{proof}

As one final note, we show how Theorem~\ref{thm:qec_main} applies to the previously-introduced error correction example from Section~\ref{sec:qecprelim}.
\begin{exam}
{\rm Recalling the CPTP map $\cl{E}$ and correctable subspace $\cl{C}_0$ from Example~\ref{exam:ec01}, consider the map defined by the following four Kraus operators:
\begin{align*}
	\pi := \Big\{ P_{\cl{C}_0}, \ \ P_{\cl{C}_1}(X \otimes I \otimes I), \ \ P_{\cl{C}_2}(I \otimes X \otimes I), \ \ P_{\cl{C}_3}(I \otimes I \otimes X) \Big\}.
\end{align*}

\noindent It is relatively simple (if not somewhat labourious) to verify that $\pi : \cl{L}(\cl{C}_0) \rightarrow \cl{L}(\cl{H})$ is in fact a representation and satisfies $MD_\pi(\cl{E}) = \cl{L}(\cl{C}_0)$, as well as the other conditions of Theorem~\ref{thm:qec_main}. The fact that $\pi = \cl{R}^\dagger$ is no coincidence; it can be seen by the proof of the implication (1) $\Rightarrow$ (3) in Theorem~\ref{thm:qec_main} that the dual of the correction operation, when restricted to $\fA_\cl{B}$, is equal to the representation mentioned in the theorem.}
\end{exam}

%%%%%%%%%%%%%%%%%%%%%%%%%%%%%%%%%%%%%%%%%%%%%%%%%%%%%%%%%%%
\section{Outlook}\label{sec:outlook}
%%%%%%%%%%%%%%%%%%%%%%%%%%%%%%%%%%%%%%%%%%%%%%%%%%%%%%%%%%%

We see this work as adding to the increasing number of connections between operator theory and quantum information. We showed that generalized multiplicative domains can be used to characterize correctable codes for CPTP maps. Various ways of characterizing generalized multiplicative domains were derived, although we required trace-preservation of the map under consideration in all cases. An intuitive reason for {\em why} the extra hypotheses on Theorem~\ref{thm:genChar} are required remains elusive.

Our proofs of Theorems~\ref{thm:genChar} and \ref{thm:genChar2} rely on operators having a finite number of eigenvalues, and our approach in general relies heavily on the representation theory of finite-dimensional C$^*$-algebras. Thus it is not clear whether our results extend to more general C$^*$-algebras. In particular, infinite-dimensional and bimodule extensions of our results would be of interest.

Some other open problems and potential lines of investigation that arise from this work from a quantum information perspective include using Theorem~\ref{thm:qec_main} to find new classes of codes for noise models relevant to quantum information processing, and further exploring generalizations of known results about bimodule maps and the multiplicative domain to this new setting.

\vspace{0.1in}

\noindent{\bf Acknowledgements.} We thank Man-Duen Choi for helpful conversations. We are also grateful to the Fields Institute for kind hospitality during its 2009 Summer Thematic Program on Mathematics in Quantum Information. N.J. was supported by an NSERC Canada Graduate
Scholarship and the University of Guelph Brock Scholarship. D.W.K.
was supported by NSERC Discovery Grant 400160, NSERC Discovery Accelerator
Supplement 400233, and Ontario Early Researcher Award 048142.

%%%%%%%%%%%%%%

\end{document}